\documentclass[12pt]{article}
\usepackage{amsfonts}
\usepackage{amsmath}
\usepackage{amssymb}
\usepackage{graphicx}
\usepackage{geometry}
\usepackage{setspace}
\usepackage{comment}

\newtheorem{theorem}{Theorem}

\newtheorem{proposition}[theorem]{Proposition}
\newtheorem{remark}[theorem]{Remark}
\newenvironment{proof}[1][Proof]{\noindent\textbf{#1.} }{\ \rule{0.5em}{0.5em}}

\newenvironment{sistema}{\left\lbrace\begin{array}{@{}l@{}}}{\end{array}\right.}

\usepackage{hyperref} 
\usepackage[numbers]{natbib}

\begin{document}

\title{Learning and Portfolio Decisions\\ for HARA Investors}

\author{Michele Longo\footnote{Universit\`{a} Cattolica del Sacro Cuore, Largo Gemelli, 1, 20123, Milano, Italy. E-mail address: michele.longo@unicatt.it} \footnote{Corresponding Author} \qquad \qquad Alessandra Mainini\footnote{Universit\`{a} Cattolica del Sacro Cuore, Largo Gemelli, 1, 20123, Milano, Italy. E-mail address: alessandra.mainini@unicatt.it}}
\maketitle

\begin{abstract}
We maximize the expected utility from terminal wealth for an HARA investor when the market price of risk is an unobservable random variable. We compute the optimal portfolio explicitly and explore the effects of learning by comparing it with the corresponding myopic policy. In particular, we show that, for a market price of risk constant in sign, the ratio between the portfolio under partial observation and its myopic counterpart increases with respect to risk tolerance. As a consequence, the absolute value of the partial observation case is larger (smaller) than the myopic one if the investor is more (less) risk tolerant than the logarithmic investor. Moreover, our explicit computations enable to study in details the so called \textit{hedging demand} induced by learning about market price of risk. 

\end{abstract}

\newpage


\section{Introduction}

This paper is a contribution to continuous-time portfolio selection problems under partial observation.
Pioneering works about partial observation in financial contexts go back to \citep{detemple86}, \citep{dothan-feldman-86} and \citep{gennotte-86} where,  within an equilibrium framework, uncertainty is introduced in a linear Gaussian setting. Since then, financial models with partial observation have been extensively researched in an increasing number of papers.
In \citep{lakner-95} the Author studies, by means of martingale methods, the optimal investment and consumption strategies in a rather general framework and provides explicit results for logarithmic and power investors in the bayesian specification for the unobservable mean rate of return. The same methodology is applied in \citep{lakner-98} to explore further the linear gaussian setting for a terminal wealth maximizer.
By using a dynamic programming approach, the linear gaussian diffusion specification for the appreciation rate is also studied in \citep{rishel99}, who considers power-type investors, and in \citep{brendle-06}. In particular, in the latter work, the Author assumes a gaussian mean-reverting process for the unobserved process and power/exponential preferences for the agents -- thus extending the model in \citep{kim-omberg-96} to the partial observation setting -- and characterizes the value function and the optimal investment policy by means of a system of ordinary differential equations, which is then used to study the loss of utility due to partial observation.
The dynamic programming setting is also used in \citep{brennan98} and \citep{xia01} to analyze, by means of numerical simulations, the qualitative and quantitative effects of learning on the portfolio allocation of a power investor when the mean return on the risky asset is a random variable with a gaussian prior. The model in \citep{brennan98} is explicitly solved in \citep{rogers-01} and used to compare the cost of uncertainty with that of being constrained to change portfolio allocations only at discrete times. In \citep{cvitanic-lazrak-martellini-zapatero-06}, a model with constant, but random, mean rates of returns is used to analyze the value of professional analyst’s recommendations.
Perhaps, the closest papers to our work are \citep{karatzas-zhao-01} and \citep{rieder-bauerle}.
In the first paper, the Authors, by using both martingale methods and dynamic programming techniques, study in detail the general bayesian case for the market price of risk for a wide class of terminal wealth investors. The paper presents explicit solutions and analyzes the cost of uncertainty by comparing the indirect utility in the partial observable setting with the corresponding complete observation quantity (i.e. the indirect utility when the market price of risk is a completely observable random variable). Nevertheless, the work does not explore the effects of uncertainty on the optimal portfolio strategy.
In \citep{rieder-bauerle}, the Authors examine the case where the unobservable stock drift rate is described by a continuous-time finite state Markov chain and investors' preferences over terminal wealth are of either power or logarithmic type. By using a dynamic programming approach, they first prove a verification theorem for the general case and then solve explicitly the special bayesian case of a constant finite-state random variable. Moreover, they compare the optimal policy under partial observation with the \textit{myopic} portfolio, that is, the portfolio obtained by first solving the investor's maximization problem by considering the market price of risk a given parameter and then replacing the latter with its conditional expected value.
The continuous-time Markov chain model is also studied, using Martingale methods, Malliavin calculus, and Montecarlo simulations, in \citep{honda-03} and \citep{haussmann-sass-04}, where more details on the optimal policies are only given for precise values of the risk-aversion parameter.

The cited works considered the classical diffusion setting for asset prices, with several specifications for the unobservable drift. Two different approaches are presented in \citep{callegaro-di-masi-runggaldier-06} and \citep{bauerle-rieder-07} where the asset prices are driven
by a pure jump process in the first case, and a jump-diffusion process in the second. Both works assume unobservable jump intensity and the second paper also study in some details, obtaining results similar to those in \citep{rieder-bauerle}, the effect of uncertainty on the optimal portfolio strategy.

Within the Wiener setting for asset prices, few works -- to the best of our knowledge, the only exceptions are \citep{brennan98}, \citep{honda-03}, \citep{rieder-bauerle}, and \citep{cvitanic-lazrak-martellini-zapatero-06} where partial (often relying on numerical simulations)  results are presented -- explore analytically the effect of uncertainty on investment decisions by comparing portfolio allocations under partial observation with the corresponding myopic counterparts.
In this paper we study a continuous-time portfolio choice problem under partial observation where investors are characterized by Hyperbolic Absolute Risk Aversion (HARA) utility functions, do not observe the stock appreciation rate nor the Brownian motion driving the price process, and maximize expected utility from terminal wealth by choosing portfolio strategies based on the stock price observations. We choose a Bayesian approach by describing the market price of risk as an unobservable random variable, independent of the driving Brownian motion, and with known prior distribution. By means of elementary filtering techniques, the problem is reduced to a complete observation (and markovian) setting and, by a dynamic programming approach, solved explicitly. Then, we investigate a number of properties of both the optimal portfolio under partial observation and the \textit{hedging demand} due to market price of risk uncertainty, defined as the difference between the portfolio under partial observation and the corresponding myopic strategy. In our main result (Theorem \ref{theorem}), we prove, by using stochastic order arguments, that the ratio between the optimal portfolios under partial observation and myopic portfolio increases with respect to the degree of absolute risk tolerance if the market price of risk is constant in sign (no matter whether positive or negative). A direct consequence of this result is that the absolute value of the optimal portfolio in the case of partial observation is larger (respectively smaller) than the myopic case if the investor is more (respectively less) risk tolerant than the logarithmic investor. This extends the results in \citep{rieder-bauerle}, Section 7 and 8, to the general bayesian case for the market price of risk representation. Moreover, in this general setting, we also examine the monotonicity of the portfolios as well as of the hedging demand with respect to the degree of risk tolerance and answer to a conjecture raised by \citep{rieder-bauerle}. In fact, in the constant relative risk aversion specification for the agent's preferences and positive market price of risk, we show that the portfolio under partial observation is increasing with respect to the degree of risk tolerance, according to the conjecture, but, contrary to it, the hedging demand is not increasing in the degree of risk tolerance. Finally, we consider in detail the gaussian case for the unobserved market price of risk and show that the constancy in sign of the market price of risk is only a sufficient condition for the optimal portfolios and the hedging demand to have the properties discussed above. We also comment on a statement in \citep{brennan98}.

The paper is organized as follows. In Section \ref{model} we set the model and provide the main steps of the analysis. In Section \ref{portfolio-comparison} we study the optimal portfolios in a number of ways and compare our results with the existing literature. Subsection \ref{subsection-gaussian-prior} considers the gaussian case. Section \ref{conclusions} discusses possible issues for future research and concludes.


\section{The investment problem}\label{model}
Time is continuous and uncertainty is described by a complete probability space $(\Omega,\mathcal F,\mathbb P)$
equipped with a filtration $\mathbb F:=\{\mathcal F_t,\;0\leq t\leq T\}$, $T<\infty$, satisfying the usual
conditions of right-continuity and $\mathbb P$-null sets augmentation, and carrying a standard Brownian motion
$\mathbb W:=\{W_t,\;0\leq t\leq T\}$. The financial market consists of two tradable assets:
a risk-free asset with constant interest rate $r\geq0$ and a risky asset whose price $S$ is governed by the stochastic differential equation
\begin{equation}
dS_t=S_t(b dt +\sigma dW_t),
\end{equation}
where $\sigma$ is a positive constant and $b$, the expected rate of return,
is a real random variable.
In this framework, if at time $t\in[0,T]$ an investor has an initial wealth $x$ and chooses a trading strategy $\pi:=\left\{\pi_s, \; t\leq s\leq T \right\}$, where $\pi_s \in \mathbb R$ represents the amount of wealth invested in the risky asset at time $s$, then his/her wealth evolves according to
\begin{equation}
dX_s=rX_sds+\sigma\pi_s(\Theta ds +dW_s),\;\;X_t=x,\;\;t\leq s\leq T,\label{wealth}
\end{equation}
where the market price of risk $\Theta:=(b-r)/\sigma$ is assumed independent of $\mathbb{W}$ and with a known prior distribution
\begin{equation}
\mu(A):=\mathbb{P}(\Theta\in A),\;\;\;\;\;\;A\in\mathcal{B}(\mathbb{R}),
\end{equation}
that satisfy $\int_{\mathbb{R}}|\theta|\,\mu(d\theta)<\infty$.
The investor does not observe $\Theta$ nor the Brownian motion $\mathbb{W}$, whereas she/he continuously observe $S$
and aims at maximizing the expected value
\begin{equation}
\mathbb{E}\left[u\left(X_T\right)\right] \label{original-functional}
\end{equation}
by choosing an investment strategy $\pi$ based on the available information at time $t$, where $u$ belongs to one of the following classes of
 Hyperbolic Absolute Risk Aversion (HARA) utility functions:
\begin{enumerate}
\item Power utilities:
\begin{equation}
u_{\gamma}(x)=\frac{1-\gamma}{\gamma}\left(\frac{\beta x}{1-\gamma}+\eta\right)^{\gamma}; \;\;\;\gamma < 1,\;\;\gamma\neq 0,\;\;\beta>0,\;\; x\in D_{\gamma},\label{power-utility}
\end{equation}
where $D_{\gamma}:=\{x \mid \beta x/(1-\gamma) +\eta >0 \}$ (the case $\eta=0$ corresponds to the Constant Relative Risk Aversion (CRRA) family).
\item Logarithmic utilities:
\begin{equation}
u_{\log}(x)= \ln (\beta x+ \eta); \;\;\;\beta>0,\;\; x\in D_{\log},\label{log-utility}
\end{equation}
where $D_{\log}:=\{x \mid \beta x +\eta >0 \}$.
\item Exponential utilities:
\begin{equation}
u_{\exp}(x)= -\exp(-\beta x); \;\;\;\beta>0,\;\; x\in D_{\exp},\label{exp-utility}
\end{equation}
where $D_{\exp}=\mathbb R$ (this is the Constant Absolute Risk Aversion (CARA) family).
\end{enumerate}
Notice that logarithmic and exponential utilities can be seen, respectively, as limit cases as $\gamma \rightarrow 0$ and $\gamma \rightarrow -\infty$ of the power utilities. In particular, the logarithmic utility is the limit, as $\gamma \rightarrow 0$, of $u_{\gamma}(x)-(1-\gamma)/\gamma$, which, being a translation of $u_{\gamma}(x)$, represents the same preferences of the latter; whereas the exponential utility is the limit, as $\gamma \rightarrow -\infty$, of $u_{\gamma}(x)$ with $\eta=1$.

\subsection{The analysis}

The analysis is mainly concerned with power utilities, for the other two classes we just provide the main results.
Let
\begin{equation}
\mathbb{F}^S:=\left\{\mathcal{F}^S_t,\;0\leq t \leq T \right\}
\end{equation}
be the $\mathbb P$-augmented filtration generated by $S$,
then, for an agent with preferences characterized by a utility function $u_{i}$  (with $i\in\{\gamma,\log,\exp\}$, see, respectively, (\ref{power-utility}), (\ref{log-utility}), and (\ref{exp-utility})), an investment strategy $\pi$ is admissible at time $t$ with initial wealth $x$, and we write $\pi\in\mathcal{A}_t$, if it is $\mathbb{F}^S$-progressively measurable, $\mathbb{E}[\int_t^T\pi_s^2ds]<\infty$, and (\ref{wealth}) admits a unique strong solution 
 $\left\{X_s,\;t\leq s \leq T \right\}$ such that $\mathbb{P}(\left\{X_s e^{r(T-s)} \in D_{i},\;t\leq s \leq T \right\})=1$ (where $D_i$ is the domain of the utility function $u_i$, $i\in\{\gamma,\log,\exp\}$).\footnote{Observe that additional conditions on the parameters and/or on the prior of $\Theta$ could be required to ensure a well-defined problem (see, for instance, Assumption 3.1 in \citep{karatzas-zhao-01} or Section 8 in \citep{rishel99}). However, here we are not concerned with this issue.}
Now, the resulting portfolio problem, which is not of complete observation type (nor markovian), can be reduced within a complete observation setting by means of standard filtering techniques (see, for instance, \citep{karatzas-zhao-01} or Section 6 in \citep{rieder-bauerle}). Let
\begin{equation}
\mu_t(A):=\mathbb P[\Theta \in A \mid \mathcal F_t^S], \ \ \ A\in \mathcal{B}(\mathbb{R}),\;\;\;0\leq t\leq T, \label{mu_t}
\end{equation}
and
\begin{equation}
\hat{\Theta}_t:=\mathbb E[\Theta\mid\mathcal F_t^S]=\int_{\mathbb{R}}\theta \, \mu_{t}(d \theta),\;\;\;0\leq t\leq T,\label{Theta_t}
\end{equation}
be, respectively, the $\Theta$'s conditional probability and expectation given $\mathcal F_t^S$. We will represent $X_t$ in (\ref{wealth}) and $\hat{\Theta}_t$ in (\ref{Theta_t}) (or, equivalently, $\mu_t$ in (\ref{mu_t})) as solutions of stochastic differential equations driven by a $\mathbb F^S$-Brownian motion.
To this aim, define the $(\mathbb P,\mathbb F^{\Theta,W})$-exponential martingale density process
\begin{equation}
Z_t:=\exp \left(-\Theta W_t-\frac{\Theta^2}{2}t \right),\;\;\;0\leq t\leq T,
\end{equation}
where  $\mathbb F^{\Theta,W}:=\{\mathcal F^{\Theta,W}_t,\;0\leq t\leq T\}$
is the $\mathbb P$-augmented filtration generated by $\mathbb W$ and $\Theta$ (notice that $\mathbb{F}^S$ is \textquotedblleft smaller\textquotedblright\ than $\mathbb{F}^{\Theta,W}$ in the sense that $\mathcal{F}_t^S \subset \mathcal{F}_t^{\Theta,W}$ for all $0\leq t\leq T$: \textit{partial observation}). Then,
by Girsanov theorem, the process $\mathbb{\tilde W}:=\{Y_t,\;0\leq t\leq T\}$, with
\begin{equation}
Y_t:=\Theta t + W_t,\;\;\;0\leq t\leq T, \label{process_Y}
\end{equation}
is a $\mathbb F^{\Theta,W}$-Brownian motion under the probability measure $\mathbb{\tilde P}$ defined by the Radon-Nikodym derivative
\begin{equation}
\frac{d\mathbb{\tilde P}}{d\mathbb P}=Z_T.\label{Radon-Nikodym}
\end{equation}
The processes $\{S_t,\;0\leq t\leq T\}$ and $\{Y_t,\;0\leq t\leq T\}$ generate the same filtration;
 hence, $\mathbb{\tilde W}$ is also a $\mathbb F^S$-Brownian motion independent of $\Theta$ under $\mathbb{\tilde P}$ and $\mathbb{\tilde P}[\Theta\in A]= \mathbb{P}[\Theta\in A]=\mu(A)$, for all $A\in\mathcal B(\mathbb R)$.
 Furthermore,
\begin{equation}
Z^{-1}_t=\exp \left(\Theta Y_t-\frac{\Theta^2}{2}t \right),\;\;\;0\leq t\leq T,
\end{equation}
is a $(\mathbb{\tilde P},\mathbb F^{\Theta,W})$-martingale
 and
\begin{equation}
\mathbb{\tilde E}[Z_T^{-1}\mid\mathcal F^S_t]=
\begin{sistema}
F(t,Y_t),\;\;0<t\leq T,\\ 1,\;\;\;\;\;\;\;\;\;\;\;\;t=0,
\end{sistema}
\end{equation}
where
\begin{equation}
F(t,y):=\int_{\mathbb R}\exp \left(\theta y-\frac{\theta^2}{2}t \right)\mu(d\theta),\;\;(t,y)\in(0,T]\times\mathbb R. \label{F(t,y)}
\end{equation}
Moreover,
\begin{equation}
\mathbb{\tilde E}[Z_T^{-1}\mathbf{1}_A(\Theta )\mid\mathcal F^S_t]=
\begin{sistema}
\int_Ae^{\theta y-\theta^2t/2}\mu(d\theta)\mid_{y=Y_t},\;\;0<t\leq T,\\ \mu(A),\;\;\;\;\;\;\;\;\;\;\;\;\;\;\;\;\;\;\;\;\;\;\;\;\;\;\;\;t=0,
\end{sistema}
\end{equation}
where $\mathbf{1}_A$ is the indicator function of $A\in\mathcal B(\mathbb R)$,
and, by Bayes' rule (see \citep{karatzas-shreve-book-91}, Lemma 5.3 p. 193),
\begin{equation}
\mu_t(A)=\frac{\mathbb{\tilde E}[Z_T^{-1}\mathbf{1}_A(\Theta )\mid\mathcal F^S_t]}{\mathbb{\tilde E}[Z_T^{-1}\mid\mathcal F^S_t]}, \; \; \; A\in \mathcal B(\mathbb R). \label{mu_t(A)}
\end{equation}
Then, the following representation for $\hat{\Theta}_t$ (see (\ref{Theta_t})) holds:  $\hat{\Theta}_0= \int_{\mathbb{R}}\theta \,\mu(d\theta)$ and\footnote{In the sequel, for a given function $\psi (x_1,x_2,x_3)$, $\psi^{\prime}_{x_i}$ denotes the first partial derivative w.r.t. $x_i$ and $\psi^{\prime\prime}_{x_i x_j}$ denotes the second partial derivative w.r.t. $x_i$, $x_j$ ($i,j=1,2,3$).}
\begin{equation}
\hat{\Theta}_t=\left. \int_{\mathbb R}\theta p(t,y,\theta)\mu(d\theta)\right|_{y=Y_t} =
\frac{F^{\prime}_y(t,Y_t)}{F(t,Y_t)},\;\;\; 0 < t \leq T,\label{theta_t}
\end{equation}
where
\begin{equation}
p(t,y,\theta):=\frac{e^{\theta y-\theta^2t/2}}{F(t,y)},\;\;\;(t,y,\theta)\in(0,T]\times\mathbb R\times\mathbb R. \label{p(t,y,theta)}
\end{equation}
Therefore, $p(t,y,\cdot)$ and
\begin{equation}
\hat{\Theta}(t,y):=\int_{\mathbb R}\theta p(t,y,\theta)\mu(d\theta),\;\;\; (t,y)\in(0,T]\times\mathbb{R},\label{theta_t-y}
\end{equation}
represent, respectively, $\Theta$'s conditional density, w.r.t. the dominating measure $\mu(d\theta)$, and $\Theta$'s conditional expected value if at time $t$ we observe $Y_t = y$.
Moreover, the representation $\hat{\Theta}_t=\hat{\Theta}(t,Y_t)$ implies that the process $\{Y_t,\;0\leq t\leq T\}$, whose differential is nothing but the \textit{observed} market price of risk since
\begin{equation}
dY_t=\frac{1}{\sigma}\left( \frac{dS_t}{S_t} -r dt\right),
\end{equation}
represents a \textit{sufficient statistics} for the estimation of $\Theta$ and, by It\^{o}'s rule,
\begin{equation}
d\hat{\Theta}_t = \hat{V}_{\Theta}(t) \left( dY_t - \hat{\Theta}_t dt \right),
\end{equation}
where
\begin{equation}
\hat{V}_{\Theta}(t):=\mathbb{E}\left[ \left(\Theta - \hat{\Theta}_t\right)^2 \mid \mathcal F^S_t \right]
\end{equation}
is the $\Theta$'s conditional variance. That is, the changes over time of the Bayes' estimator $\hat{\Theta}_t$ are opposite proportional, by a factor equal to the conditional variance, to the observed \textit{error}
\begin{equation}
\hat{\Theta}_t dt - dY_t.
\end{equation}
In other words, the estimation of $\Theta$ follows an adaptive learning process. Furthermore, the representation $\hat{p}_t(\theta):=p(t,Y_t,\theta)$ (see (\ref{p(t,y,theta)})) for $\Theta$'s conditional density and It\^{o}'s rule yield
\begin{equation}
d\hat{p}_t(\theta)=\hat{p}_t(\theta)(\theta -\hat{\Theta}_t)(dY_t -\hat{\Theta}_t dt),
\end{equation}
where the process $Y_t -\int_0^t\hat{\Theta}_s ds$, known as \textit{innovation} process in filtering theory, is a $\mathbb{F}^S$-brownian motion under the original probability measure $\mathbb{P}$.

Given the previous definitions, for $i\in \{\gamma, \exp, \log \}$ and $(t,x,y)\in [0,T]\times \mathbb R\times\mathbb R$ such that $xe^{r(T-t)}\in D_{i}$ (see (\ref{power-utility}), (\ref{log-utility}), and (\ref{exp-utility})), the original investment problem (\ref{wealth})-(\ref{original-functional}) is equivalent to the following markovian problem: maximize
\begin{equation}
\mathbb{\tilde E}^{t,x,y}\left[F(T,Y_T)u_{i}\left(X_T\right)\right],
\end{equation}
over all $\pi\in\mathcal A_t$,  subject to
\begin{equation}
\begin{sistema}\medskip dX_s=rX_sds+\pi_s\sigma d\tilde W_s,\;\;X_t=x\\
dY_s=d\tilde W_s,\;\;\;\;\;\;\;\;\;\;\;\;\;\;\;\;\;\;\;\;\;\;Y_t=y,
\end{sistema}\label{x-y}
\end{equation}
where $\mathbb{\tilde E}^{t,x,y}$ denotes the conditional expectation, w.r.t. $\mathbb{\tilde P}$, given $X_t=x$ and $Y_t=y$. 
The value function is
\begin{equation}
\hat{\mathcal{V}}_{i}(t,x,y):=\sup_{\pi\in\mathcal A_t}\mathbb{\tilde E}^{t,x,y}\left[F(T,Y_T)u_{i}\left(X_T\right)\right],\label{value function-2}
\end{equation}
 and, under appropriate regularity conditions, the dynamic programming principle 
 yields the following Hamilton-Jacobi-Bellman equation for $\hat{\mathcal{V}}_{i}$:
\begin{equation}
f^{\prime}_t+\sup_{\pi \in \mathbb R}\left\{\frac{1}{2}\sigma^2\pi^2f^{\prime\prime}_{xx}+\sigma\pi f^{\prime\prime}_{xy}+\frac{1}{2}f^{\prime\prime}_{yy}+rxf^{\prime}_x \right\}=0, \label{HJB}
\end{equation}
for all $(t,x,y)\in[0,T)\times  \mathbb R \times\mathbb R$ s.t. $xe^{r(T-t)}\in D_{i}$, with boundary condition
\begin{equation}
f(T,x,y)=F(T,y)u_{i}(x),\quad(x,y)\in  D_{i}\times\mathbb R,  \label{HJB-boundary}
\end{equation}
$i\in \{\gamma, \exp, \log \}$. Assuming $f^{\prime\prime}_{xx}<0$, the maximization on the LHS of (\ref{HJB}) gives the following optimal value for $\pi$:
\begin{equation}
\pi =-\frac{f^{\prime\prime}_{xy}}{\sigma f^{\prime\prime}_{xx}}.\label{control-policy}
\end{equation}
By substituting (\ref{control-policy}) back into (\ref{HJB}), the latter reduces to
\begin{equation}
f^{\prime}_t-\frac{1}{2}\frac{\left(f^{\prime\prime}_{xy}\right)^2}{f^{\prime\prime}_{xx}}
+\frac{1}{2}f^{\prime\prime}_{yy}+rxf^{\prime}_x=0. \label{HJB_2}
\end{equation}

\textit{Power utilities}. If $\gamma<1, \gamma \neq 0$, we try a solution of the form
\begin{equation}
f(t,x,y)=u_{\gamma}\left( x e^{r(T-t)} \right)h(t,y)^{1-\gamma}.\label{guess}
\end{equation}
By computing the partial derivatives and substituting into (\ref{HJB_2}), it is found that (\ref{guess}) is a solution of (\ref{HJB_2})-(\ref{HJB-boundary}) if and only if $h$ solves
\begin{equation}
h^{\prime}_t+\frac{1}{2}h^{\prime\prime}_{yy}=0, \;\;\;(t,y)\in[0,T)\times\mathbb R ,\label{heat-equation-1}
\end{equation}
with boundary condition  $h(T,y)= F(T,y)^{1/(1-\gamma)}$, $y\in\mathbb R$.
An application of Feynman-Kac formula yields the following solution for (\ref{heat-equation-1}):
\begin{equation}
\hat{h}(t,y;\gamma):=
\left\{
\begin{array}
[l]{l}%
 \int_{\mathbb R}  F(T,y+z)^{1/(1-\gamma)}\varphi_{T-t}(z)dz,\;\;\;(t,y)\in[0,T)\times\mathbb R \\
F(T,y)^{1/(1-\gamma)},\;\;\;\;\;\;\;\;\;\;\;\;\;\;\;\;\;\;\;\;\;\;\;\;\;\;\;\;\;(t,y)\in\{T\}\times\mathbb R,
\end{array}
\right. \label{h(t,y)}
\end{equation}
where $\varphi_{T-t}(\cdot)$ is the normal density with mean $0$ and variance $T-t$.
Then, a standard verification argument 
will enable us to prove the following representation for the value function (\ref{value function-2}) with $i=\gamma$:
 \begin{equation}
 \hat{\mathcal{V}}_{\gamma}(t,x,y)=u_{\gamma}\left( x e^{r(T-t)} \right)\hat{h}(t,y;\gamma)^{1-\gamma}, \label{value-function-gamma}
\end{equation}
for all $(t,x,y)\in[0,T]\times  \mathbb R \times\mathbb R$ s.t. $xe^{r(T-t)}\in D_{\gamma}$, and the optimality of the markov policy (see (\ref{control-policy}))
 \begin{equation}
\hat{\pi}_{\gamma}(t,x,y):=\left(\frac{x}{\sigma (1-\gamma)}+ \frac{\eta e^{-r(T-t)}}{\sigma \beta }\right)\int_{\mathbb R}\hat{\Theta}(T,y+z) q(t,y,z;\gamma)dz,\label{optimal-portfolio-partial}
\end{equation}
$(t,x,y)\in[0,T)\times  \mathbb R \times\mathbb R$ s.t. $xe^{r(T-t)}\in D_{\gamma}$,  where $\hat{\Theta}(\cdot,\cdot)$ is defined in (\ref{theta_t-y}) and
\begin{equation}
q(t,y,z;\gamma):=\frac{F(T,y+z)^{1/(1-\gamma)}\varphi_{T-t}(z)}{\int_{\mathbb R}\,F(T,y+z)^{1/(1-\gamma)}\varphi_{T-t}(z)dz}, \label{q(t,y,z,gamma)}
\end{equation}
$(t,y,z)\in[0,T)\times \mathbb R\times\mathbb R$. Notice that $q(t,y, \cdot;\gamma)$ is a probability density function w.r.t. the Lebesgue measure for all $(t,y)\in[0,T)\times\mathbb R $ and for any $\gamma < 1$.

\begin{remark}\label{remark-segno-primo-fattore-portafoglio}
Let $\{X^{\ast}_s,\;t\leq s\leq T\}$ be the optimal wealth process in (\ref{x-y}) once we substitute the optimal feedback policy (\ref{optimal-portfolio-partial}). Then, by It\^{o}'s rule,
\[
\tilde{\mathbb{P}}(\left\{X^{\ast}_s e^{r(T-s)} \in D_{\gamma},\;t\leq s \leq T \right\})=1,
 \]
 if the initial condition $(t,x)$ is such that $xe^{r(T-t)}\in D_{\gamma}$ (cf. \citep{merton-71}). Therefore, the expression on the RHS of (\ref{value-function-gamma}) is well-defined and the first factor on the RHS of (\ref{optimal-portfolio-partial}) is positive along the optimal wealth process.
\end{remark}

\textit{Logarithmic and Exponential utilities}. Similar computations show that the value functions for the logarithmic and exponential utilities (see (\ref{log-utility}) and (\ref{exp-utility})) are, respectively,
\begin{equation}
\begin{aligned}
 \hat{\mathcal{V}}_{\log}(t,x,y)= & \; F(t,y)u_{\log}(xe^{r(T-t)})-F(t,y)\ln F(t,y)  \\
 + & \int_{\mathbb R}F(T,z+y)\ln F(T,z+y)\varphi_{T-t}(z)dz   \label{value-function-log}
\end{aligned}
\end{equation}
and
\begin{equation}
 \hat{\mathcal{V}}_{\exp}(t,x,y)= u_{\exp}(xe^{r(T-t)})\exp{\left(\int_{\mathbb R}\ln F(T,z+y)\varphi_{T-t}(z)dz\right)}. \label{value-function-exp}
\end{equation}
The corresponding optimal portfolios are, respectively,
\begin{equation}
\hat{\pi}_{\log}(t,x,y) := \left(\frac{x}{\sigma}+ \frac{\eta e^{-r(T-t)}}{ \sigma \beta}\right)\hat{\Theta}(t,y)\label{optimal-portfolio-partial-log}
\end{equation}
and
\begin{equation}
\hat{\pi}_{\exp}(t,x,y) := \frac{e^{-r(T-t)}}{\sigma \beta} \,\int_{\mathbb R}\hat{\Theta}(T,y+z)\varphi_{T-t}(z)dz,\label{optimal-portfolio-partial-exp}
\end{equation}
where, again, $\hat{\Theta}(\cdot,\cdot)$ is defined in (\ref{theta_t-y}).
The following Proposition proves that the logarithmic and exponential portfolios (\ref{optimal-portfolio-partial-log}) and (\ref{optimal-portfolio-partial-exp}) are the limits, respectively as $\gamma \rightarrow 0$ and as $\gamma \rightarrow -\infty$, of the power portfolio (\ref{optimal-portfolio-partial}).

\begin{proposition}\label{proposition-1}
(i) For any $(t,x,y)\in[0,T)\times  \mathbb R \times\mathbb R$ such that $xe^{r(T-t)}\in D_{\log}$, we have
\begin{equation}
\lim_{\gamma\to 0}\hat{\pi}_{\gamma}(t,x,y)=\hat{\pi}_{\log}(t,x,y).\label{log_portfolio}
\end{equation}
(ii) For each $(t,x,y)\in[0,T)\times D_{\exp} \times\mathbb R$, the exponential portfolio is the limit, as $\gamma \rightarrow -\infty$, of the power portfolio with $\eta =1$, that is
\begin{equation}
\lim_{\gamma\to -\infty} \left( \left. \hat{\pi}_{\gamma}(t,x,y)\right|_{\eta=1} \right)=\hat{\pi}_{\exp}(t,x,y).\label{exp_portfolio}
\end{equation}
\end{proposition}
\begin{proof}
(i) The result follows by observing that
\begin{align*}
\lim_{\gamma\to 0}\int_{\mathbb R}\hat{\Theta}(T,y+z) q(t,y,z;\gamma)dz = & \int_{\mathbb R}\hat{\Theta}(T,y+z) \frac{F(T,y+z)\varphi_{T-t}(z)}{\int_{\mathbb R}\,F(T,y+z)\varphi_{T-t}(z)dz} \,dz \\
= &\; \frac{\int_{\mathbb R}\,\int_{\mathbb R}\,\theta e^{\theta(z+y)-\theta^2T/2}\mu(d\theta)\varphi_{T-t}(z)dz}{\int_{\mathbb R}\,\int_{\mathbb R}\, e^{\theta(z+y)-\theta^2T/2}\mu(d\theta)\varphi_{T-t}(z)dz} \\
= &\; \frac{\int_{\mathbb R}\,\int_{\mathbb R}\,\theta e^{\theta(z+y)-\theta^2T/2}\varphi_{T-t}(z)dz\mu(d\theta)}{\int_{\mathbb R}\,\int_{\mathbb R}\, e^{\theta(z+y)-\theta^2T/2}\varphi_{T-t}(z)dz\mu(d\theta)}\\
= & \;\frac{\int_{\mathbb R}\,\theta e^{\theta y-\theta^2t/2}\mu(d\theta)}{\int_{\mathbb R}\,e^{\theta y-\theta^2t/2}\mu(d\theta)} =\; \hat{\Theta}(t,y),
\end{align*}
where the first equality is a consequence of the dominated convergence theorem, and the third makes use of Fubini's theorem.\\
(ii) Again, the result follows by observing that, by the dominated convergence theorem,
\[
\lim_{\gamma\to -\infty}\int_{\mathbb R}\hat{\Theta}(T,y+z) q(t,y,z;\gamma)dz = \int_{\mathbb R}\hat{\Theta}(T,y+z) \varphi_{T-t}(z)\,dz.
\]
\end{proof}
\begin{remark} \label{remark_theta_t-y_representation}
Item (i) of previous proposition proves the following representation for $\hat{\Theta}(t,y)$ in (\ref{theta_t-y}):
\begin{equation}
\hat{\Theta}(t,y)
=\int_{\mathbb R}\hat{\Theta}(T,y+z) q(t,y,z;0)dz, \label{theta_t-y_representation-eq}
\end{equation}
for all $(t,y)\in[0,T)\times \mathbb R$.
\end{remark}

\section{Learning and portfolio decision}\label{portfolio-comparison}

In this section we analyse in more details the results of the previous section and show a number of properties of the optimal portfolios. In particular, we explore the effect of learning on portfolio decisions by comparing the portfolios (\ref{optimal-portfolio-partial}), (\ref{optimal-portfolio-partial-log}), and (\ref{optimal-portfolio-partial-exp}) with the so called \textit{myopic} strategies. We say that an investor is myopic if she/he refrains from learning about the true distribution of $\Theta$ in the sense that, having observed $Y_t=y$ at time $t$, he behaves as the expected value $\hat{\Theta}(t,y)$ were certain. In other words, the investors first solve the maximization problem by considering the market price of risk a parameter -- say, $\Theta \equiv \theta$, $\theta \in \mathbb R$ -- and then replace the latter with its conditional expected value in the optimal investment policy. The optimal Merton's portfolios associated to a market price of risk constant and equal to $\theta$ are the following: for the power utility investor,
\begin{equation}
\pi_{\gamma}^{Merton}(t,x,y;\theta):=\left(\frac{x}{\sigma(1-\gamma)}+ \frac{\eta e^{-r(T-t)}}{\sigma \beta}\right)\theta ,\label{optimal-portfolio-merton}
\end{equation}
$\gamma<1, \gamma \neq 0$; for the logarithmic and exponential investor we have, respectively,
\begin{equation}
\pi^{Merton}_{\log}(t,x,y;\theta):=\left(\frac{x}{\sigma}+ \frac{\eta e^{-r(T-t)}}{\sigma \beta }\right)\theta \label{optimal-portfolio-merton-log}
\end{equation}
and
\begin{equation}
\pi^{Merton}_{\exp}(t,x,y;\theta): = \frac{e^{-r(T-t)}}{ \sigma\beta } \,\theta. \label{optimal-portfolio-merton-exp}
\end{equation}
Notice that also in the Merton's model the logarithmic and exponential portfolios are the limits, respectively as $\gamma \rightarrow 0$ and as $\gamma \rightarrow -\infty$, of the power portfolio.
\begin{remark}\label{remark-certainty-equivalence-principle}
For a logarithmic investor the portfolio under partial observation (\ref{optimal-portfolio-partial-log}) is obtained by substituting into Merton's portfolio (\ref{optimal-portfolio-merton-log}) the parameter $\theta$ with its conditional expected value $\hat{\Theta}(t,y)$ (see (\ref{theta_t-y})), that is,
\begin{equation}
\hat{\pi}_{\log}(t,x,y)=\pi^{Merton}_{\log}(t,x,y;\hat{\Theta}(t,y)).
\end{equation}
This is the so-called \emph{certainty equivalence principle} and does not hold for other types of preferences (see \citep{kuwana-95}). In other words, the logarithmic investor is myopic.
\end{remark}
Now, the myopic strategies for the power, logarithmic and exponential investors are thus defined:
\begin{equation}
\pi_{\gamma}^{m}(t,x,y):=\pi_{\gamma}^{Merton}(t,x,y;\hat{\Theta}(t,y))=\left(\frac{x}{\sigma(1-\gamma)}+ \frac{\eta e^{-r(T-t)}}{\sigma \beta}\right)\hat{\Theta}(t,y) , \label{optimal-portfolio-merton-con-Theta}
\end{equation}
\begin{equation}
\pi_{\log}^{m}(t,x,y):=\pi_{\log}^{Merton}(t,x,y;\hat{\Theta}(t,y))=\left(\frac{x}{\sigma}+ \frac{\eta e^{-r(T-t)}}{\sigma \beta }\right)\hat{\Theta}(t,y), \label{optimal-portfolio-merton-log-con-Theta}
\end{equation}
and
\begin{equation}
\pi_{\exp}^{m}(t,x,y):=\pi_{\exp}^{Merton}(t,x,y;\hat{\Theta}(t,y))= \frac{e^{-r(T-t)}}{ \sigma\beta } \,\hat{\Theta}(t,y). \label{optimal-portfolio-merton-exp-con-Theta}
\end{equation}

The following proposition shows that all HARA investors become myopic as time elapses.
\begin{proposition}\label{proposition-limit-T}
For each class of HARA utility functions (\ref{power-utility}), (\ref{log-utility}), and (\ref{exp-utility}), the portfolio under partial observation tends, as $t \rightarrow T$, to the corresponding myopic portfolio. That is, for each $i\in \{\gamma, \exp, \log \}$,
\begin{equation}
\lim_{t\to T}\hat{\pi}_{i}(t,x,y)=\pi^{m}_{i}(T,x,y), \;\;\; (x,y)\in D_{i} \times \mathbb R . \label{portfolio-at-T}
\end{equation}
\end{proposition}
\begin{proof}
There is nothing to prove for the logarithmic utility. For the power utility, by recalling (\ref{optimal-portfolio-partial}) and (\ref{optimal-portfolio-merton-con-Theta}), it suffices to show that, for any $\gamma < 1$ and any $y\in \mathbb R$,
\[
\lim_{t\to T} \;\int_{\mathbb R}\hat{\Theta}(T,y+z) q(t,y,z;\gamma)dz = \hat{\Theta}(T,y).
\]
This is established by using the representation
\[
\int_{\mathbb R}\hat{\Theta}(T,y+z) q(t,y,z;\gamma)dz =\frac{\mathbb E \left[\hat{\Theta}(T,y+W_{T-t})F(T,y+W_{T-t})^{1/(1-\gamma)}\right]}{\mathbb E \left[F(T,y+W_{T-t})^{1/(1-\gamma)}\right]},
\]
together with the continuity of the functions $\hat{\Theta}$ and $F$ (see, respectively, (\ref{theta_t-y}) and (\ref{F(t,y)})), and observing that $W_{T-t}\to 0$, almost surely, as $t\to T$. For the exponential utility the proof proceeds similarly.
\end{proof}

The next result and its implications generalize Theorem 6 in \citep{rieder-bauerle} to the case of a market price of risk described by a general random variable.
\begin{theorem}\label{theorem}
Assume the random variable $\Theta$ is constant in sign $\mu-a.s.$ (i.e., either $\Theta$ takes only positive values or only negative values\footnote{Presumably, the case of negative values is not the most realistic situation from a financial point of view, albeit it cannot be excluded.}). Then, for any fixed $(t,y)\in[0,T)\times\mathbb R$ and $x \in D_{\gamma}$, the ratio
\begin{equation}
\frac{\hat{\pi}_{\gamma}(t,x,y)}{\pi^{m}_{\gamma}(t,x,y)}, \;\;\; \gamma <1,
\end{equation}
is positive, has limit $1$ as $\gamma \rightarrow 0$, and is increasing with respect to $\gamma$.
\end{theorem}

\begin{proof}
The assertions that $\hat{\pi}_{\gamma}/\pi^{m}_{\gamma}$ is positive and has limit equal to $1$, as $\gamma \rightarrow 0$, follow by recalling definition (\ref{theta_t-y}), representation (\ref{theta_t-y_representation-eq}) and observing that \begin{equation}
\frac{\hat{\pi}_{\gamma}(t,x,y)}{\pi^{m}_{\gamma}(t,x,y)}= \frac{\int_{\mathbb R}\hat{\Theta}(T,y+z) q(t,y,z;\gamma)dz}{\hat{\Theta}(t,y)}= \frac{\int_{\mathbb R}\hat{\Theta}(T,y+z) q(t,y,z;\gamma)dz}{\int_{\mathbb R}\hat{\Theta}(T,y+z) q(t,y,z;0)dz}, \label{theorem-proof-1}
\end{equation}

To prove the monotonicity of $\hat{\pi}_{\gamma}/\pi^{m}_{\gamma}$ w.r.t. $\gamma$, it suffices to show that, for any fixed $(t,y)\in[0,T)\times\mathbb R $,
\begin{equation}
f(t,y,\gamma):= \int_{\mathbb R}\hat{\Theta}(T,y+z) q(t,y,z;\gamma)dz, \;\;\; \gamma <1,
\end{equation}
as a function of $\gamma$, is increasing if $\Theta$ is positive and decreasing if $\Theta$ is negative. To this end, let $\gamma_1 <\gamma_2 <1$ and consider the ratio of densities (w.r.t. $z$)
\begin{align*}
\frac{ q(t,y,z;\gamma_2)}{q(t,y,z;\gamma_1)}
= &\; \frac{\int_{\mathbb R}\,F(T,y+z)^{1/(1-\gamma_1)}\varphi_{T-t}(z)dz}{\int_{\mathbb R}\,F(T,y+z)^{1/(1-\gamma_2)}\varphi_{T-t}(z)dz}
\times \frac{F(T,y+z)^{1/(1-\gamma_2)}}{F(T,y+z)^{1/(1-\gamma_1)}}\\
= &\; C(t,y,\gamma_1,\gamma_2)\times F(T,y+z)^{k}
\end{align*}
where
\[
C(t,y,\gamma_1,\gamma_2):=\frac{\int_{\mathbb R}\,F(T,y+z)^{1/(1-\gamma_1)}\varphi_{T-t}(z)dz}{\int_{\mathbb R}\,F(T,y+z)^{1/(1-\gamma_2)}\varphi_{T-t}(z)dz}>0
\]
and
\[
k:=\frac{1}{1-\gamma_2}-\frac{1}{1-\gamma_1}>0.
\]
Hence, $q(t,y,z;\gamma_2)/q(t,y,z;\gamma_1)$ is increasing w.r.t. $z$ if and only if $F(T,y+z)$ is increasing w.r.t. $z$ and (see definition (\ref{F(t,y)})) $F(T,y+z)$, as a function of $z$, is increasing if $\Theta \geq0$ and decreasing if $\Theta \leq0$. Therefore, the random variable with density $q(t,y,\cdot;\gamma_2)$ dominates, with respect to the likelihood ratio order and hence in the usual stochastic order, the random variable with density $q(t,y,\cdot;\gamma_1)$ if $\Theta \geq0$, whereas the dominance is reversed if $\Theta \leq0$ (see \citep{bauerle-rieder-book-2011}, Appendix B.3). If $\hat{\Theta}(T,y+z)$ were increasing w.r.t. $z$, then the stochastic dominance just proved would imply
\begin{align*}
f(t,y,\gamma_1)= & \;\int_{\mathbb R}\hat{\Theta}(T,y+z) q(t,y,z;\gamma_1)dz \\
\leq & \; \int_{\mathbb R}\hat{\Theta}(T,y+z) q(t,y,z;\gamma_2)dz=f(t,y,\gamma_2),
\end{align*}
if $\Theta \geq0$, and the opposite inequality if $\Theta \leq0$. That is, the claimed monotonicity for $f$.

To complete the proof we need to show that $\hat{\Theta}(T,y+z)$ is increasing w.r.t. $z$. To see this, observe that for any $t\in(0,T]$ and any $y_1<y_2$, the ratio of densities, w.r.t. the dominating measure $\mu(d\theta)$ (see (\ref{p(t,y,theta)})),
\[
\frac{p(t,y_2,\theta)}{p(t,y_1,\theta)}=e^{\theta(y_2 -y_1)}D(t,y_1,y_2),
\]
where $D(t,y_1,y_2):=F(t,y_1)/F(t,y_2) >0$  (see (\ref{F(t,y)})), is increasing w.r.t. $\theta$. Hence, by the likelihood ratio order, we have
\[
\hat{\Theta}(t,y_1)=\int_{\mathbb R}\,\theta p(t,y_1,\theta) \mu(d\theta)\leq\int_{\mathbb R}\,\theta p(t,y_2,\theta) \mu(d\theta)=\hat{\Theta}(t,y_2),
\]
for any $t\in(0,T]$. In particular, $\hat{\Theta}(T,y+z)$ is increasing w.r.t. $z$.
\end{proof}

\begin{remark}\label{remark-abs-value-portfolio}
\emph{
If the market price of risk $\Theta$ is constant in sign then Theorem \ref{theorem} implies $\operatorname{sign}\left( \hat{\pi}_{\gamma} \right) = \operatorname{sign}\left(\pi^{m}_{\gamma} \right) =\operatorname{sign}\left( \Theta \right)$ and
\begin{equation} \label{monotonicity-abs-values}
\text{if}\; \gamma >0,\; \text{then} \; \left\vert \hat{\pi}_{\gamma}\right\vert \geq\left\vert \pi_{\gamma}
^{m}\right\vert, \; \text{and if}\; \gamma < 0, \;\text{then} \;
\left\vert \hat{\pi}_{\gamma}\right\vert \leq\left\vert \pi_{\gamma}%
^{m}\right\vert .
\end{equation}
The intuition for the previous result is as follows. Since agents can be either short or long on the market, it is the magnitude of the market price of risk rather than its sign that makes the risky asset attractive for \textit{aggressive} power investors (i.e., power investor with $0 < \gamma < 1$). That is why, for $0 < \gamma < 1$, non-negative market price of risk implies $\hat{\pi}_{\gamma}\geq \pi_{\gamma}^{m}$ and non-positive market price of risk implies $\hat{\pi}_{\gamma}\leq\pi_{\gamma}^{m}$.
\textit{Conservative} power investors (i.e. power investor with $\gamma < 0$) behave in the opposite direction (cf. \citep{brennan98}, see also \citep{kim-omberg-96}, Section 3). We further notice that $\left\vert \hat{\pi}_{\gamma}\right\vert$ itself is increasing w.r.t. $\gamma$; that is, the higher is the risk tolerance, the more is the wealth invested into the risky asset.}
\end{remark}

\begin{remark}\label{remark-limits}
\emph{
Let $\theta_1,\theta_2 \in \mathbb R$ and assume $0<\theta_1 \leq \Theta \leq \theta_2$, $\mu-a.s.$, then $\theta_1 \leq \hat{\Theta}(t,y) \leq \theta_2$   (see \ref{theta_t-y}), and
\begin{equation}
\frac{\theta_1}{\theta_2} \leq \frac{\hat{\pi}_{\gamma}(t,x,y)}{\pi^{m}_{\gamma}(t,x,y)} \leq \frac{\theta_2}{\theta_1}, \;\; \gamma <1,
\end{equation}
for all $(t,x,y) \in [0,T)\times D_{\gamma} \times \mathbb R$ (see \ref{theorem-proof-1}). Therefore,
\begin{equation}
\frac{\theta_1}{\theta_2} \leq \lim_{\gamma \rightarrow -\infty} \frac{\hat{\pi}_{\gamma}(t,x,y)}{\pi^{m}_{\gamma}(t,x,y)} \leq 1 \leq \lim_{\gamma \rightarrow 1^-} \frac{\hat{\pi}_{\gamma}(t,x,y)}{\pi^{m}_{\gamma}(t,x,y)} \leq \frac{\theta_2}{\theta_1}.
\end{equation}
Notice that the limits, as $\gamma \rightarrow -\infty$ and as $\gamma \rightarrow 1^-$, of $\hat{\pi}_{\gamma}/\pi^{m}$ depend, in general, on the distribution of $\Theta$.}
\end{remark}

\begin{remark}\label{remark-monotonicity-headging-demand}
\emph{
Theorem \ref{theorem} enables us to investigate the monotonicity of $\hat{\pi}_{\gamma}$ and $\hat{\pi}_{\gamma}-\pi_{\gamma}^{m}$ with respect to $\gamma$. In particular, the latter difference, in light of the decomposition
\begin{equation}
\hat{\pi}_{\gamma}=\pi_{\gamma}^{m} + (\hat{\pi}_{\gamma}-\pi_{\gamma}^{m}),
\end{equation}
might be considered a sort of hedging demand against market price of risk uncertainty (cf. \citep{brennan98} and \citep{honda-03}, where the analysis is performed numerically, see also \citep{cvitanic-lazrak-martellini-zapatero-06}). Without loss of generality, we present the case of a positive market price of risk, the results are reversed in the case of a negative market price of risk. Moreover, with respect to the parameter $\eta$, two cases are of interest: $\eta=0$ and $\eta=1$ (since all the other cases can be either treated similarly or reduced to one of them). For the sake of comparison, we also examine $\pi_{\gamma}^{m}$.}

\emph{
(a) Case $\eta=0$.
Here, we extend the results in \citep{rieder-bauerle}, Sections 7 and 8, to the case of a general random variable for the market price of risk.  We have $D_{\gamma}=\{x\in \mathbb R \mid x>0\}$,
\begin{equation}
\hat{\pi}_{\gamma}(t,x,y)=\frac{x}{\sigma (1-\gamma)}\int_{\mathbb R}\hat{\Theta}(T,y+z) q(t,y,z;\gamma)dz,
\end{equation}
and, by (\ref{theta_t-y_representation-eq}),
\begin{equation}
\hat{\pi}_{\gamma}(t,x,y)-\pi_{\gamma}^{m}(t,x,y)=\frac{x}{\sigma (1-\gamma)}\int_{\mathbb R}\hat{\Theta}(T,y+z)(q(t,y,z;\gamma)-q(t,y,z;0))dz. \label{portfolio-difference_eta=0}
\end{equation}
Notice that the first factor in the LHS of (\ref{portfolio-difference_eta=0}) is positive and increasing w.r.t. $\gamma$, and that the second factor is also increasing w.r.t. $\gamma$, but positive for $\gamma >0$ and negative for $\gamma <0$ (see Theorem \ref{theorem}'s proof).
Then, for any $(t,x,y)\in[0,T)\times  (0,\infty) \times\mathbb R$:
\begin{itemize}
\item[(i)] $\hat{\pi}_{\gamma}$ and $\pi_{\gamma}^{m}$ are both increasing w.r.t. $\gamma$,
\begin{equation}
\lim_{\gamma \rightarrow -\infty} \hat{\pi}_{\gamma} = \lim_{\gamma \rightarrow -\infty}\pi_{\gamma}^{m}=0^+, \; \;\; \text{and} \;\;\;\lim_{\gamma \rightarrow 1^{-}} \hat{\pi}_{\gamma} = \lim_{\gamma \rightarrow 1^{-}}\pi_{\gamma}^{m}=+\infty,
\end{equation}
provided, for the limit as $\gamma \rightarrow 1^{-}$, that the value function and the optimal portfolios exist in a left neighborhood of $\gamma =1$ (see Remark \ref{remark-existence-gaussian} ahead for this issue in the gaussian setting). Hence, our analysis confirms the  conjecture in \citep{rieder-bauerle}, p. 376, on the monotonicity of $\hat{\pi}_{\gamma}$, w.r.t. $\gamma$.
\item[(ii)] The difference $\hat{\pi}_{\gamma}-\pi_{\gamma}^{m}$ is such that
\begin{equation}
\lim_{\gamma \rightarrow -\infty} \hat{\pi}_{\gamma}-\pi_{\gamma}^{m} = \lim_{\gamma \rightarrow 0}\hat{\pi}_{\gamma}-\pi_{\gamma}^{m}=0 \;\; \text{and} \;\;
\lim_{\gamma \rightarrow 1^{-}}\hat{\pi}_{\gamma}-\pi_{\gamma}^{m}=+\infty,
\end{equation}
where, as before, the last limit relies on the existence of the portfolios in a left neighborhood of $\gamma =1$. Furthermore, there exists $-\infty < \underline{\gamma} \leq 0$ such that $\hat{\pi}_{\gamma}-\pi_{\gamma}^{m}$ is increasing for all  $\gamma \geq \underline{\gamma}$, albeit such a difference is not increasing for all $\gamma \in (-\infty, 1)$, contrary to what has been conjectured in \citep{rieder-bauerle}, p. 376. Notice that in \citep{rieder-bauerle} the portfolio represents the fraction of wealth invested in the risky asset rather than the amount, but this does not affect the validity of our results also in their setting.
\end{itemize}}

\emph{
(b) Case $\eta=1$.  Then,
\begin{equation}
\hat{\pi}_{\gamma}(t,x,y)=\left(\frac{x}{\sigma(1-\gamma)}+ \frac{e^{-r(T-t)}}{\sigma \beta}\right)\int_{\mathbb R}\hat{\Theta}(T,y+z) q(t,y,z;\gamma)dz,
\end{equation}
and
\begin{equation}
\begin{aligned}
\hat{\pi}_{\gamma}(t,x,y)-\pi_{\gamma}^{m}(t,x,y)= &\left(\frac{x}{\sigma(1-\gamma)}+ \frac{ e^{-r(T-t)}}{\sigma \beta}\right)\\
 \times & \int_{\mathbb R}\hat{\Theta}(T,y+z)(q(t,y,z;\gamma)-q(t,y,z;0))dz. \label{portfolio-difference_eta=1}
\end{aligned}
\end{equation}
By Remark \ref{remark-segno-primo-fattore-portafoglio}, the first factor in the LHS of (\ref{portfolio-difference_eta=1}) is positive, it is increasing w.r.t. $\gamma$, if $x\geq 0$, and (strictly) decreasing, if $x<0$.
\begin{itemize}
\item[--] If $x\geq 0$, the behaviors of $\hat{\pi}_{\gamma}$ and $\hat{\pi}_{\gamma}-\pi_{\gamma}^{m}$ are identical to case 3(a) except for the limits as $\gamma \rightarrow -\infty$, which are computed according to Proposition \ref{proposition-1}, and the monotonicity of $\hat{\pi}_{\gamma}-\pi_{\gamma}^{m}$, where, depending on the distribution of $\Theta$ and/or the value for $x$, we may have $ \underline{\gamma} =-\infty $ (i.e. the difference might be increasing in its entire domain).
\item[--] If $x<0$, then $\gamma \leq 1 + x\beta e^{r(T-t)} < 1$, the limits of $\hat{\pi}_{\gamma}$ and $\hat{\pi}_{\gamma}-\pi_{\gamma}^{m}$, as $\gamma \rightarrow -\infty$, are as before, and there exists $ \bar{\gamma}  < 1 + x\beta e^{r(T-t)} $ such that $\hat{\pi}_{\gamma}-\pi_{\gamma}^{m}$ is increasing for all  $\gamma \leq \bar{\gamma}$.
\end{itemize}
}
\end{remark}

\begin{remark}\label{remark-relative-hedging-demand}
\emph{
The quantity
\begin{equation}
\frac{\hat{\pi}_{\gamma}-\pi_{\gamma}^{m}}{\pi_{\gamma}^{m}}=\frac{\hat{\pi}_{\gamma}}{\pi_{\gamma}^{m}}-1
\end{equation}
could be defined as relative hedging demand for uncertainty, whose properties can be easily deduced from Theorem \ref{theorem} and remarks thereafter.
}
\end{remark}

If the market price of risk $\Theta$ is not constant in sign then, in general, Theorem \ref{theorem} and its Remarks
no longer hold -- an exception is the case of a gaussian prior (see Subsection \ref{subsection-gaussian-prior} ahead) -- as the numerical example in \citep{rieder-bauerle}, Figure 2 page 376, shows.

For the sake of completeness, we present also the portfolio ratios for logarithmic and exponential investors. By Remark \ref{remark-certainty-equivalence-principle}, $\hat{\pi}_{\log}(t,x,y)/\pi^{m}_{\log}(t,x,y)$ is identically equal to $1$. Whereas, for exponential preferences, if $\Theta$ is constant in sign, the stochastic order technique used in the proof of Theorem \ref{theorem}, item (ii), proves
\begin{equation}
0 \leq \frac{\hat{\pi}_{\exp}(t,x,y)}{\pi^{m}_{\exp}(t,x,y)} \leq 1, \label{portfolio-ratio-exp}
\end{equation}
 for all $(t,x,y) \in [0,T)\times D_{\exp} \times \mathbb R$. This could also be established by observing that Theorem \ref{theorem}, item (ii), and Proposition \ref{proposition-1} imply
\begin{equation}
\frac{\hat{\pi}_{\exp}(t,x,y)}{\pi^{m}_{\exp}(t,x,y)}= \lim_{\gamma \rightarrow -\infty} \frac{\hat{\pi}_{\gamma}(t,x,y)}{\pi^{m}_{\gamma}(t,x,y)} \leq 1,
\end{equation}
 for all $(t,x,y) \in [0,T)\times D_{\exp} \times \mathbb R$.

\subsection{The Gaussian case}\label{subsection-gaussian-prior}

If the prior of $\Theta$ is normally distributed with mean $m$ and variance $v^2$, then the integrals in (\ref{F(t,y)}), (\ref{theta_t-y}), and (\ref{optimal-portfolio-partial}) may be computed explicitly to yield, respectively,
\begin{equation}
F(t,y)=\frac{1}{\sqrt{1+v^2t}} \exp\left(\frac{(m+v^2y)^2}{2v^2(1+v^2t)}-\frac{m^2}{2v^2}\right),
\end{equation}
\begin{equation}
\hat{\Theta}(t,y)=\frac{m+v^{2}y}{1+v^{2}t},
\end{equation}
and
\begin{equation}
\hat{\pi}_{\gamma}(t,x,y)= \left(\frac{x}{\sigma(1-\gamma)}+ \frac{\eta e^{-r(T-t)}}{\sigma \beta}\right) \hat{\Theta}(t,y)\frac{(1-\gamma)(1+v^{2}t)}{(1-\gamma)-\gamma v^2T+v^2t}.   \label{optimal-poltfolio-gaussian}
\end{equation}
This case is considered in \citep{brennan98} and explicitly solved in \citep{rogers-01}. In particular, in \citep{brennan98}, the Author assumes standard power preferences and the portfolio represents the fraction of wealth invested into the risky asset; hence, the optimal portfolio in that setting (equation (9), page 300) is given by (\ref{optimal-poltfolio-gaussian}) with $\eta=0$ and $x=1$.\footnote{In Brennan's notation, the portfolio $\alpha$ numerically approximated in Table I, page 302, has the following analytic form:
$\alpha=\left(  m_0-r\right)  / (\sigma^{2}(  1-\gamma -\gamma\left(  v_0 /\sigma \right)  ^{2}T))  $.}
\begin{remark}\label{remark-existence-gaussian}
Existence and finiteness of portfolio (\ref{optimal-poltfolio-gaussian}) 
are guaranteed under the condition $(1-\gamma)-\gamma v^2T>0$ (cf. \citep{karatzas-zhao-01}, Remark 5.6). Hence, the analysis of $\hat{\pi}_{\gamma}$, $\hat{\pi}_{\gamma}/\pi^{m}_{\gamma}$, as well as $\hat{\pi}_{\gamma}-\pi^{m}_{\gamma}$, w.r.t. $\gamma$, makes sense only for $\gamma < 1/(1+v^2T) <1$.
\end{remark}
Now, for all $\gamma < 1/(1+v^2T) <1$, the ratio
\begin{equation}
 \frac{\hat{\pi}_{\gamma}(t,x,y)}{\pi^{m}_{\gamma}(t,x,y)}=\frac{(1-\gamma)(1+tv^2)}{(1-\gamma)-\gamma v^2T+v^2t},
 \end{equation}
is positive, has limit $1$ as $\gamma \rightarrow 0$, and is increasing with respect to $\gamma$. Therefore, with reference to Theorem \ref{theorem}, the constancy in sign of the market price of risk $\Theta$ represents only a sufficient condition for the claimed properties of the portfolio ratio $\hat{\pi}_{\gamma}/\pi^{m}_{\gamma}$.
Furthermore, similarly to Remark \ref{remark-abs-value-portfolio}, we have
\begin{equation}
\operatorname{sign}\left( \hat{\pi}_{\gamma}(t,x,y) \right) = \operatorname{sign}\left(\pi^{m}_{\gamma}(t,x,y) \right) =\operatorname{sign}\left( \hat{\Theta}(t,y) \right)
\end{equation}
and (\ref{monotonicity-abs-values}) holds. Here, our analysis disagrees with \citep{brennan98} where the Author seems to claim that if $\gamma >0$ then $\hat{\pi}_{\gamma} \geq \pi_{\gamma}^{m}$, and if $\gamma < 0$, then $\hat{\pi}_{\gamma} \leq \pi_{\gamma}^{m}$, at least if we should take literally a proposition on the Abstract at page 295 (see also page 300). But, as our study shows, this is true if and only if $\hat{\Theta}(t,y) \geq 0$ and, in a gaussian setting, we may very well have $\hat{\Theta}(t,y) \leq 0$, in which case the order between the portfolios is reversed. Hence, the precise statement on page 295 should read as follows: \lq\lq... the possibility of learning about the mean return on the risky asset induces the investor to take a larger or smaller position, \textit{in absolute value}, in the risky asset than she would if there were no learning, the direction of the effect depending on whether the investor is more or less risk tolerant than the logarithmic investor.\rq\rq\

The hedging demand in the gaussian setting becomes
\begin{equation}
\hat{\pi}_{\gamma}(t,x,y)-\pi^{m}_{\gamma}(t,x,y)=\left(\frac{x}{\sigma(1-\gamma)}+\frac{\eta e^{-r(T-t)}}{\sigma\beta}\right)\hat{\Theta}(t,y)
\frac{\gamma v^2(T-t)}{(1-\gamma)-\gamma v^2T+v^2t},
\end{equation}
and the analysis as well as the conclusions about monotonicity of $\hat{\pi}_{\gamma}$ and $\hat{\pi}_{\gamma}-\pi_{\gamma}^{m}$ with respect to $\gamma$ are, depending on the sign of $\hat{\Theta}(t,y)$, similar to those of Remark \ref{remark-monotonicity-headging-demand}. Moreover, the relative hedging demand (see Remark \ref{remark-relative-hedging-demand}) takes the following very simple form:
\begin{equation}
\frac{\hat{\pi}_{\gamma}-\pi^{m}_{\gamma}}{\pi^{m}_{\gamma}}= \frac{\gamma v^2(T-t)}{(1-\gamma)-\gamma v^2T+v^2t},
\end{equation}
which dramatically explodes as $\gamma \rightarrow (1+v^2t)/(1+v^2T)<1$, meaning that the effects of uncertainty may be very relevant in the portfolio allocation for certain power investors.

The optimal portfolios for logarithmic and exponential investors in the gaussian setting can be obtained by computing the limits of (\ref{optimal-poltfolio-gaussian}), respectively, as $\gamma \rightarrow 0$ and as $\gamma \rightarrow -\infty$.

\section{Conclusions}\label{conclusions}
In this paper we study the problem of maximizing the expected utility from terminal wealth for HARA investors when the market price of risk is represented by an unobservable random variable with known general prior distribution. We solve the problem and compare the optimal portfolio with the corresponding myopic strategy, that is the portfolio obtained by first solving the investor's maximization problem by considering the market price of risk a given parameter and then replacing the latter with its conditional expected value. We show that, if the prior for the market price of risk is constant in sign, then the ratio between the portfolio under partial observation and the corresponding myopic case is increasing with respect to the degree of risk tolerance (to be precise, the ratio is increasing with respect to the sensitivity of absolute risk tolerance to wealth).
This has a number of consequences such as: 1) The absolute value of portfolio under partial observation is increasing with respect to risk tolerance and it is larger than the myopic counterpart if and only if the investor is more risk tolerant than the logarithmic investor (Remark \ref{remark-abs-value-portfolio}); 2) The absolute hedging demand against uncertainty is increasing at least for degree of risk tolerance sufficiently high (Remark \ref{remark-monotonicity-headging-demand}); 3) The relative hedging demand is increasing for all possible degree of risk tolerance (Remark \ref{remark-relative-hedging-demand}).

The model presented in this paper can be further studied in a number of ways. First, a question of interest is what other types of prior distributions guarantee the properties established in Theorem \ref{theorem} for the portfolios ratio.   
A second issue worth of investigation is the loss of utility for a myopic investor (a sort of cost of myopia). 
Finally, it would be interesting to compare the optimal portfolio under partial observation with the case where the market price of risk is a completely observable random variable (some works have already begun the exploration of this issue, see, for instance, \citep{karatzas-zhao-01} and \citep{brendle-06}. All these are subjects for future research.

\bibliography{biblio}

\begin{thebibliography}{10}

\bibitem{bauerle-rieder-07}
N.~B\"auerle and U.~Rieder.
\newblock Portfolio optimization with jumps and anobservable intensity process.
\newblock {\em Mathematical Finance}, 17:205--224, 2007.

\bibitem{bauerle-rieder-book-2011}
N.~B\"auerle and U.~Rieder.
\newblock {\em Markov Decision Processes with Application to Finance}.
\newblock Springer-Verlag, Berlin Heidelberg, 2011.

\bibitem{brendle-06}
S.~Brendle.
\newblock Portfolio selection under incomplete information.
\newblock {\em Stochastic Processes and their Applications}, 106:701--723,
  2006.

\bibitem{brennan98}
M.J. Brennan.
\newblock The role of learning in dynamic portfolio decisions.
\newblock {\em European Finance Review}, 1:295--306, 1998.

\bibitem{callegaro-di-masi-runggaldier-06}
G.~Callegaro, G.~Di~Masi, and W.~Runggaldier.
\newblock Portfolio optimization in discontinuous markets under incomplete
  ionformation.
\newblock {\em Asia-Pac Financ Mark}, 13:373–394, 2006.

\bibitem{cvitanic-lazrak-martellini-zapatero-06}
J.~Cvitani\'{c}, A.~Lazrak, L.~Martellini, and F.~Zapatero.
\newblock Dynamic portfolio choice with parameter uncertainty and the economic
  value of analysts’ recommendations.
\newblock {\em The Review of Financial Studies}, 19:1113--1156, 2006.

\bibitem{detemple86}
J.~Detemple.
\newblock Asset pricing in a production economy with incomplete information.
\newblock {\em The Journal of Finance}, 41:383--391, 1986.

\bibitem{dothan-feldman-86}
M.~Dothan and D.~Feldman.
\newblock Equilibrium interest rates and multiperiod bonds in a partially
  observable economy.
\newblock {\em The Journal of Finance}, 41:369--382, 1986.

\bibitem{gennotte-86}
G.~Gennotte.
\newblock Optimal portfolio choice under incomplete information.
\newblock {\em The Journal of Finance}, 41:733--746, 1986.

\bibitem{honda-03}
T.~Honda.
\newblock Optimal portfolio choice for unobservable and regime-switching mean
  returns.
\newblock {\em Journal of Economic Dynamics and Control}, 28:45--78, 2003.

\bibitem{karatzas-shreve-book-91}
I.~Karatzas and S.~Shreve.
\newblock {\em Brownian motion and stochastic calculus}.
\newblock Springer, New York, 1991.

\bibitem{karatzas-zhao-01}
I.~Karatzas and X.~Zhao.
\newblock Bayesian adaptive portfolio optimization.
\newblock In Cvitani\'c J. Musiela~M. Jouini, E., editor, {\em Option Pricing,
  Interest Rates and Risk Management}, pages 632--669. Cambridge University
  Press, 2001.

\bibitem{kim-omberg-96}
T.S. Kim and E.~Omberg.
\newblock Dynamic nonmyopic portfolio behavior.
\newblock {\em The Review of Financial Studies}, 9:141--161, 1996.

\bibitem{kuwana-95}
Y.~Kuwana.
\newblock Certainty equivalence and logarithmic utilities in
  consumption/investment problems.
\newblock {\em Mathematical Finance}, 5:297--310, 1995.

\bibitem{lakner-95}
P.~Lakner.
\newblock Utility maximization with partial information.
\newblock {\em Stochastic Processes and their Applications}, 56:247--273, 1995.

\bibitem{lakner-98}
P.~Lakner.
\newblock Optimal trading strategy for an investor: The case of partial
  information.
\newblock {\em Stochastic Processes and their Applications}, 76:77--97, 1998.

\bibitem{merton-71}
R.~Merton.
\newblock Optimum consumption and portfolio rules in a continuous-time model.
\newblock {\em Journal of Economic Theory}, 3:373--413, 1971.
\newblock Correction: (1973) 6, 213-214.

\bibitem{rieder-bauerle}
U.~Rieder and N.~B\"auerle.
\newblock Portfolio optimization with unobservable markov-modulated drift
  process.
\newblock {\em Journal of Applied Probability}, 42:362--378, 2005.

\bibitem{rishel99}
R.~Rishel.
\newblock Optimal portfolio management with partial observation and power
  utility function.
\newblock In Yin~G. McEneany, W. and Q.~Zhang, editors, {\em Stochastic
  Analysis, Control, Optimization and Applications: Volume in Honor of W. H.
  Fleming}, pages 605--620. Birkh\"auser Verlag, Boston, 1999.

\bibitem{rogers-01}
L.C.G. Rogers.
\newblock The relaxed investor and parameter uncertainty.
\newblock {\em Finance and Stochastics}, 5:131--154, 2001.

\bibitem{haussmann-sass-04}
J.~Sass and U.G. Haussmann.
\newblock Optimizing the terminal wealth under partial information: The drift
  process as a continuous time markov chain.
\newblock {\em Finance and Stochastics}, 8:553--577, 2004.

\bibitem{xia01}
Y.~Xia.
\newblock Learning about predictability: The effects of parameter uncertainty
  on dynamic asset allocation.
\newblock {\em The Journal of Finance}, 56:205--246, 2001.

\end{thebibliography}
\bibliographystyle{plain}

\end{document}